\documentclass[11pt]{article}

\newcommand{\mypara}[1]{\smallskip\noindent\textbf{#1.}}  

\usepackage[text={6.5in,9in},centering]{geometry}
\usepackage[cmex10]{amsmath}
\usepackage{enumerate}
\usepackage{graphicx}
\usepackage{amsfonts}
\usepackage{amsthm}
\usepackage{url}
\usepackage{gensymb}
\usepackage{pdfpages}

\newtheorem{proposition}{Proposition}
\newtheorem{lemma}{Lemma}
\newtheorem{corollary}{Corollary}
\newtheorem{theorem}{Theorem}

\newtheorem{claim}{Claim}
\newtheorem{fact}{Fact}

\def\calE{{\mathcal{E}}}
\def\calN{{\mathcal{N}}}

\def\calS{{\mathcal{S}}}
\def\calI{{\mathcal{I}}}

\def\Pro{{\mathbb{P}}}
\def\Ex{{\mathbb{E}}}

\def\cD{{\mathcal{D}}}
\def\cG{{\mathcal{G}}}
\def\cR{{\mathcal{F}}}
\def\GS{\calS^\cG}
\def\GI{\calI^\cG}
\def\SS{\calS^{\cD}}
\def\SI{\calI^{\cD}}
\def\RS{\calS^\cR}
\def\RI{\calI^\cR}
\def\GSIR{{SIR^\cG}}
\def\SSIR{{SIR^{\cD}}}
\def\OPTG{OPT^\cG}
\def\OPTD{OPT^{\cD}}
\def\OPTS{OPT^{\cD}}
\def\OPTR{OPT^\cR}
\def\affG{a^\cG}
\def\affS{a^{\cD}}

\newcommand{\prob}[1]{\textsc{#1}}
\newcommand{\capacity}{\prob{LinkCapacity}}

\begin{document}

\begin{titlepage}

\title{Wireless Link Capacity under Shadowing and Fading}

\author{
  Magn\'us M. Halld\'orsson
  \qquad
  Tigran Tonoyan \\ \\
  ICE-TCS, School of Computer Science \\
  Reykjavik University \\
  \url{mmh@ru.is, ttonoyan@gmail.com}
}

\maketitle
\thispagestyle{empty}

\begin{abstract}
  We consider the following basic \emph{link capacity} (a.k.a., one-shot scheduling) problem in wireless networks: 
Given a set of communication links, find a maximum subset of links that can successfully transmit simultaneously.
Good performance guarantees are known only for deterministic models, such as the physical model with geometric (log-distance) pathloss. 
We treat this problem under stochastic shadowing under general distributions, bound the effects of shadowing on optimal capacity,
and derive constant approximation algorithms.
We also consider temporal fading under Rayleigh distribution, 
and show that it affects non-fading solutions only by a constant-factor.
These can be combined into a constant approximation link capacity algorithm under both time-invariant shadowing and
temporal fading.
\end{abstract}

\end{titlepage}

\section{Introduction}

Efficient use of networks requires attention to the \emph{scheduling} of the communication.
Successful reception of the intended signal requires attention to the \emph{interference} from other simultaneous transmissions,
and both are affected crucially by the vagaries in the \emph{propagation of signals} through media.
We aim to understand the fundamental capacity question of \emph{how much communication can coexist},
and the related algorithmic aspect of \emph{how to select large sets of successfully coexisting links}.
The hope is to capture the reality of signal propagation, while maintaining 
the fullest generality: arbitrary instances, and minimal distributional assumptions.

The basic property of radio-wave signals is that they attenuate as they travel. In free space, the attenuation (or ``pathloss'') grows with the \emph{square} of the distance. In any other setting, there are obstacles, walls, ceilings and/or the ground, 
in which the waves can go through complex transforms: reflection, refraction (or shadowing), scattering, and  diffraction. The signal received by a receiver is generally a combination of the multiple paths that it can travel, that are phase-shifted, resulting in patterns of constructive and destructive interference. The general term for variation of the received strength of the signal from the free-space expectation is \emph{fading}.

The fading of signals can be a function of time, location, frequency and other parameters, of which we primarily focus on the first two.
A distinction is often made between \emph{large-scale} fading, the effects of larger objects like buildings and trees, and
the \emph{small-scale} fading at the wavelength scale caused by multiple signal-paths.
We primarily distinguish between \emph{temporal fading}, that varies randomly within 
the time frame of communication, and \emph{shadowing}, that is viewed as 
invariant within the time horizon of consideration.
Temporal fading is typically experienced at the small-scale as a combination of multi-path propagation and movement or other environment changes.

The most common way of modeling true fading is \emph{stochastic fading}.
To each point in space-time, we associate a random variable drawn from a distribution. 
Typically, this is given by a distribution in the logarithmic dBm scale, so on an absolute scale the distributions are exponential.
There is a general understanding that \emph{log-normal shadowing} (LNS), which is Gaussian on the dBm scale, is the most faithful
approximation known of medium-large scale fading or even all atemporal fading \cite{zamalloa2007,chen2011,dezfouli2015modeling}. 
Empirical models often add variations depending on the environment, 
the heights of the sender/receiver from the ground, and whether there is a line-of-sight (e.g., \cite{mondal2015}).
The most prominent among the many models proposed for small-scale and temporal fading
are the ones of Rayleigh and Rice \cite{tse2005}, with the former (latter) best suited when there is (is no) line-of-sight, respectively.
The Rayleigh distribution mathematically captures the case when the signal is highly scattered and equally likely to arrive at any angle.
Though probabilistic models are known to be far from perfect, 
they are generally understood to be highly useful for providing insight into wireless systems,
and certainly more so than the free-space model alone.

Stochastic fading is the norm in generational models, such as for simulation purposes. 
For instance, LNS is built into the popular NS-3 simulator.
It is also commonly featured in stochastic analysis, e.g. \cite{Liu05eurasip}.
Worst-case analysis of algorithms has, however, nearly always involved deterministic models,
either the geometric free-space model, or extensions to more general metric spaces, e.g., (cite various).
One might expect such analysis to treat similarly arbitrary or ``any-case'' fading, but that quickly invokes the 
ugly specter of computational intractability \cite{GoussevskaiaHW14}.

\mypara{Problems and setting}
In the \emph{Link Capacity} problem, we are given a set $L$ of links, each of which is a sender-receiver pair of nodes on the plane.
We seek a maximum \emph{feasible} subset of links in $L$, where a set is feasible in the physical (or SINR) model if, for each link,
the strength of the signal at the receiver is $\beta$ times larger than total strengths of the interferences from the other links.
We consider \emph{arbitrary}/any-case positions of links, aiming for algorithms with good performance guarantees,
as well as characterizations of optimal solutions.

We treat Link Capacity in extensions of the standard physical model to stochastic fading.
We separate the fading into \emph{temporal} and \emph{atemporal} (or time-invariant) aspects, 
which we refer to as temporal fading and shadowing, respectively.
We generally assume independence across space in time-invariant distributions and across time in the temporal distributions.
This is a simplification, aimed to tackle most pronounced aspects; where possible, we relax the independence assumptions, sometimes allowing for arbitrary (worst-case) values.
Observe that the two forms can be arbitrarily correlated: the temporal results hold under \emph{arbitrary} time-invariant fading.

\mypara{Our results}
We give a comprehensive treatment of link capacity under stochastic fading models.
We give constant-factor approximation algorithms for both time-invariant and temporal stochastic models.
These are complementary and can be multiplexed into algorithms for both types of fading.

For (time-invariant) shadowing, we allow for essentially any reasonable stochastic distribution.
We show that shadowing never decreases the optimal link capacity (up to a constant factor), but can significantly increase it,
where the prototypical case is that of co-located links. We give algorithms for general instances, that achieve a constant factor approximation, assuming length diversity is constant.

For temporal fading, we treat arbitrary instances that can have \emph{arbitrary pathloss/shadowing}.
We show that algorithms that ignore the temporal fading given by Rayleigh distribution achieve a constant factor approximation.
The links can additionally involve weights and can be of arbitrary length distribution.

Besides the specific results obtained, our study leaves us with a few implications that may be of general 
utility for algorithm and protocol designers.
One such lesson is that to achieve good performance for shadowing,
\begin{quote}
\emph{\noindent
algorithms can concentrate on the signal strengths of the links},
\end{quote}
and can largely ignore the strength of the interference between the links. Another useful lesson is that
\begin{quote}
\emph{\noindent
algorithms can base decisions on time-invariant shadowing alone},
\end{quote}
since the temporal fading will even out.

These appear to be the first any-case analysis of scheduling problems in general stochastic models.
In particular, ours appears to be the first treatment of approximation algorithm for scheduling problems under shadowing or time-invariant fading.

\mypara{Related Work}
Gupta and Kumar \cite{guptakumar} introduced the physical model, which corresponds to our setting with no fading.
Their work spawned off a large number of studies on ``scaling laws'' regarding throughput capacity in instances with stochastic input distributions. 
First algorithms with performance guarantees in the physical model were given by Moscibroda and Wattenhofer \cite{MoscibrodaW06}.
Constant approximation for the Link Capacity problem were given for uniform power \cite{GoussevskaiaHW14}, 
linear power \cite{FKV09,Tigran11a}, fixed power assignments \cite{SODA11}, and arbitrary power control \cite{KesselheimSODA11}. 
This was extended to a distributed setting \cite{Dinitz2010,Eyjo11}, 
admission control in cognitive radio \cite{HM12cognitive},
link rates \cite{KesselheimESA12}, multiple channels \cite{us:podc14,Wan16},
spectrum auction \cite{HoeferK15},
changing spectrum availability \cite{dams2013sleeping}, 
and MIMO \cite{Wan14-mimo}.
NP-hardness was established in \cite{goussevskaiacomplexity}.
Numerous works on heuristics are known, as well as exponential time exact algorithms (e.g., \cite{shi2011maximizing}).

The Link Capacity problem has been fundamental to various other scheduling problems, appearing as a key subroutine for
shortest link schedule \cite{wan2011,GoussevskaiaHW14,us:talg12},
maximum multiflow \cite{Wan09,WiOpt14},
weighted link capacity \cite{wan2011,KesselheimESA12},
and capacity region stability \cite{CISS12,kesselheimStability}. 

Numerous experimental results have indicated that simplistic range-based models of wireless reception
are insufficient, e.g., \cite{ganesan2002,kotz2004experimental,zamalloa2007}. 
Significant experimental literature exists that lends support for stochastic models \cite{nikookar1993}, especially log-normal shadowing,
e.g., \cite{zamalloa2007,chen2011,dezfouli2015modeling}.
Analytic results on stochastic fading are generally coupled with stochastic assumptions on the inputs, 
such as point processes in stochastic geometry \cite{Liu05eurasip,haenggi2009,WWW15}. 
Most stochastic fading models though do not lend themselves to closed-form formulation; Rayleigh fading is a rare exception \cite{cardieri}. 
Log-normal shadowing has been shown to result in better connectivity \cite{stuedi2005connectivity,muetze2008understanding} and throughput capacity \cite{stuedi2007log}, but this may be artifact of the i.i.d.\ assumption \cite{agrawal2009correlated}.

The only work on Link Capacity with any-case instances in fading models is by Dams et al.\ \cite{dams2015}, 
who showed that temporal Rayleigh fading does not significantly affect the performance of {\capacity}
algorithms, incurring only a $O(\log^* n)$-factor increase in performance for link capacity algorithms.  We improve this
here to a constant factor. Rayleigh fading has also been considered in distributed algorithms for local broadcast \cite{wang2016}.

The non-geometric aspects of signal propagation have been modeled non-stochastically in various ways.
One simple mechanism is to vary the pathloss constant $\alpha$ \cite{guptakumar}.
A more general approach is to view the variation as deforming the plane into a general metric space \cite{FKRV09,SODA11}.
Also, the pairwise pathlosses can be obtained directly from measurements, inducing a quasi-metric space \cite{us:podc14}.
All of these, however, lead to very weak performance guarantees in the presence of the huge signal propagation
variations that are seen in practice (although some of that can be ameliorated by identifying parameters with better
behavior, like ``inductive independence'' \cite{HoeferK15}).

\section{Models and Formulations}

\subsection{Communication Model}

The main object of our consideration is a set $L$ of \emph{communication links}, numbered from $1$ to $n=|L|$. Each link $i\in L$ represents a unit-demand communication request between a sender node $s_i$ and a receiver node $r_i$, both point-size wireless nodes located on the plane. 

We assume the links all work in the same channel,  and all (sender) nodes use the same transmission power level $P$ (unless stated otherwise).
We consider the following basic question, which is called the {\capacity} problem: what is the maximum number of links in $L$ that can successfully communicate in a single time slot? We will refer to a set of links that can  successfully communicate in a single time slot as \emph{feasible}.

When a subset $S$ of links transmit at the same time, a given link $i$ will succeed if its signal (the power of the transmission of $s_i$ when measured at $r_i$) is larger than $\beta$ times the total (sum) interference from other transmissions, where $\beta\ge 1$ is a threshold parameter, and the interference of link $j$ on link $i$ is the power of transmission of $s_j$ when measured at $r_i$. We will denote by $\calS_i$ the received signal power of link $i$ and by $\calI_{ji}$ the interference of link $j$ on link $i$. In this notation, link $i$ transmits successfully if\footnote{In general, there should also be a Gaussian noise term in the success condition, which is omitted for simplicity of exposition. It may be noted that in expectation, the success of only a fraction of links will be affected by the noise.}
\begin{equation}\label{E:sir}
SIR(S,i)=\frac{\calS_i}{\sum_{j\in S\setminus i}\calI_{ji}} > \beta.
\end{equation}

\subsection{Geometric Path-Loss}

The \emph{Geometric Path-Loss model} or GPL for short,
defines the received signal strength between nodes $u$ and $v$ as $P/d(u,v)^\alpha$, where $P$ is the power used by the sender $u$, $\alpha>2$ is the \emph{path-loss exponent} and $d$ denotes the Euclidean distance. In particular,   the signal strength/power of a link $i$ and the interference of a link $j$ on link $i$ are, respectively, \[\GS_i=P/l_i^\alpha\text{ and }\GI_{ji}=P/d(s_j,r_i)^\alpha,\] 
where $l_i=d(s_i,r_i)$ denotes the \emph{length} of link $i$ and $d(s_j,r_i)$ is the distance from the sender node of link $j$ to the receiver node of link $i$.

If the links in a set $S$ transmit simultaneously, the formula determining the success of  the transmission on link $i$ is similar to (\ref{E:sir}), but we will use the slightly modified notation $\GSIR(S,i)>\beta$ to indicate that GPL model is considered.

\subsection{Shadowing}

One of the effects that GPL ignores (or models only by appropriate change of the exponent $\alpha$), is signal obstruction by objects, or \emph{shadowing}. In generic networks shadowing is often modeled by a \emph{Stochastic Shadowing model}, or SS for short, such as the \emph{Log-Normal Shadowing model}, or LNS for short. In this case, there is a parametrized probability distribution $\cD$, such that the signal strength $\SS_i$ of a link $i$ at $r_i$ is assumed to have been sampled from the distribution $\cD$ and $\Ex[\SS_i]=\GS_i$, and similarly, for any two links $i,j$, the interference $\SI_{ij}$ is sampled from $\cD$ and $\Ex[\SI_{ij}]=\GI_{ij}$.  We assume that 
signals and interferences do not change in time (due to shadowing), at least during the time period when {\capacity} needs to be solved. 
In this model too, signal reception is characterized by the signal to interference ratio, but we will use the notation $\SSIR(S,i) > \beta$ to indicate that SS model is considered.

We shall be assuming independence among the random variables. This may lead to artifacts that are contrary to experience.
It is nevertheless valuable to examine closely this case that might be considered the most extreme.

\subsection{Temporal Fading}

Another effect that is not described by the models above is the temporal variations in the signal,
due to a combination of movement (of either transceivers or people/objects in the environment) and the scattered multipath components of the signal.  We will concentrate on \emph{Rayleigh fading}, where the signal power $\RS_i$ is distributed according to an exponential distribution with mean $\calS_i$, i.e. $\Ex[\RS_i]=\calS_i$, and similarly, the interference power $\RI_{ij}$ is distributed according to an \emph{exponential distribution} with mean $\calI_{ij}$, i.e., $\Ex[\RI_{ij}]=\calI_{ij}$, where $\calS_i$ and $\calI_{ij}$ are the signal and interference values, not necessarily from SS or GPL.  Again, the success of transmission is described by the signal to interference ratio. Note, however, that in this case the success is probabilistic: the same set of links can be feasible in one time slot and non-feasible in another.

\subsection{Computational Aspects} There is a striking difference between GPL and shadowing on one side, and temporal fading on the other side, from the computational point of view. This difference stems from the spatial nature of GPL and shadowing, and the time-variant nature of temporal fading. In the former case, an algorithm can be assumed to have access to, e.g., the signal strengths of links, which could be obtained by measurements. This, however, is impossible or impractical under temporal fading, which forces the algorithms to be probabilistic and base the actions solely on the expected values of signal strengths (w.r.t.\ fading distribution), and the performance ratio of algorithms is measured accordingly (see Sec.~\ref{S:rayleigh} for details).

\subsection{Technical Preliminaries}
Throughout this text, by ``constants'' we will mean fixed values, independent of the network size and topology (e.g. distances). Some examples are parameters $\alpha,\beta$, and the constants under big O notation.
\paragraph{Affectance.} In order to describe feasibility of a set of links with arbitrary signal and interference values, we will use the notion of \emph{affectance}, which is more convenient than (but equivalent to) the signal to noise ratio. For two links $i,j$ we let $a_{i}(j)=\frac{1}{SIR(i,j)} = \frac{\calI_{ij}}{\calS_j}$ and extend this definition to subsets: If $S\subseteq L$ is a set of links, then $a_S(j)=\sum_{i\in S\setminus j}a_i(j)$ and $a_i(S)=\sum_{j\in S\setminus i} a_i(j)$. Then, a set $S$ of links is \emph{feasible} if and only if $a_S(j)< 1/\beta$ holds for every link $j\in S$. When considering a particular SS distribution $\cD$ or GPL we will use superscripts $\cD$ and $\cG$ respectively, as before. We will use the following result of~\cite{HB15}, which shows that feasibility is robust with respect to the threshold value $\beta$.
\begin{lemma}\label{L:robustness}
If a set $S$ of links and number $\beta'>0$ are such that $a_S(i)\le 1/\beta'$ for each link $i\in S$, then $S$ can be partitioned into at most $\lceil 2\beta/\beta' \rceil$ feasible subsets.
\end{lemma}

\paragraph{Smooth Shadowing Distributions.} We will use quantiles of an SS distribution. Consider an SS distribution $\cD$ and assume links use uniform power assignment. For a probability $p\in (0,1)$ and a link $i$ with $\Ex[\SS_i]=\GS_i$, let $\bar\calS^p_i$ denote the $1-p$-quantile, i.e., a value $x$ such that $\Pr[\SS_i > x]=p$.
For a given number $p\in (0,1)$, the distribution $\cD$ is called \emph{$p$-smooth}, if there is a constant $c>0$ such that $\bar\calS^p_i \ge c\cdot \GS_i$ holds for each link $i$.
Note that all major distributions used to model stochastic shadowing satisfy such a smoothness condition.

\section{Comparing Shadowing to GPL}

We start by comparing the optimal solutions of {\capacity} under SS and GPL for any given set $L$ of links. However, a particular instance drawn from an SS distribution $\cD$ is arguably not informative and can be hard to solve. Instead, we will be more interested in the gap between GPL optimum and a ``typical'' SS optimum, in the sense of expectation. 

We denote by $\OPTG(L)$ the size (number of links) of the optimal solution to {\capacity} for a set $L$ of links under GPL.  Similarly, we denote by $\OPTD(L)$ the size of the optimal solution to {\capacity} under SS model with distribution $\cD$. 
We assume that for every link $i$, the variables $\SS_i$ and $\SI_{ij}$ (for all $j$) are independent, unless specified otherwise.

 We will compare the expected value  $\Ex[\OPTD(L)]$  with $\OPTG(L)$, for a given set $L$ of links, where the expectation is over the distributions of  random variables  $\SS_i$ and $\SI_{ji}$ for all links $i,j$. In particular, we prove that the expected SS optimum is never worse than a constant factor of the GPL optimum. On the other hand, due to the presence of links with high signal strength that can appear as a result of shadowing, the capacity can considerably increase.

\subsection{SS Does not Decrease Capacity}

First, we show that $\Ex[\OPTS(L)] = \Omega(\OPTG(L))$, i.e., a typical optimum under SS is not worse than the optimum under GPL.

\begin{theorem}
Let $\cD$ be a $p$-smooth SS distribution with a constant $p>0$, and let $L$ be any set of links. Then $\Ex[\OPTS(L)]=\Omega(\OPTG(L))$.
\end{theorem}
\begin{proof}
Let $S$ be a maximum cardinality subset of $L$ that is feasible under GPL, and let us fix a link $i\in S$. Recall that $\GS_i=\Ex[\SS_i]$ and, by additivity of expectation, $\sum_{j\in S}\GI_{ji}=\Ex[\sum_{j\in S}\SI_{ji}]$.
Also, by smoothness assumption, there are constants $c>0$ and $p\in (0,1)$ such that $\bar\calS^p_i \ge c\cdot \GS_i$ and $\Pr[\SS_i > \bar\calS^p_i]=p$. On the other hand, by Markov's inequality, 
\[
\Pr\left[\sum_{j\in S\setminus i}\SI_{ji} \le 2 \sum_{j\in S\setminus i}\GI_{ji}\right] \ge 1/2.
\]
 Recall that we assumed that  the random variable $\SS_i$ is independent from $\SI_{ji}$ for each $i,j\in L$. Thus, we have that with probability at least $p/2$, both $\SS_i > \bar\calS^p_i>c\GS_i$ and $\sum_{j\in S\setminus i}\SI_{ji} \le 2 \sum_{j\in S\setminus i}\GI_{ji}$ hold, implying that $\SS_i > \frac{\beta c}{2} \sum_{j\in S\setminus i}\SI_{ji}$. By additivity of expectation, it follows that the expected size of a subset $S''$ of $S$ with $\SS_i > \frac{\beta c}{2} \sum_{j\in S\setminus i}\SI_{ji}$ for each link $i$ is at least $\frac{p}{2}\cdot|S|$. On the other hand, such a set $S''$ can be partitioned into at most $4/c$ feasible (under SS) subsets, by Lemma~\ref{L:robustness}. This implies that $\Ex[\OPTS(L)] \ge \frac{pc}{8}\cdot \OPTG(L)$.
\end{proof}

In the particular case of Log-Normal Shadowing, even independence is not necessary for the result above to hold, as shown in th next theorem. 
 We use the standard notation $X\sim \ln\calN(\mu,\sigma^2)$ and $Y\sim \calN(\mu,\sigma^2)$ to denote log-normally and normally distributed random variables, respectively. In Log-Normal Shadowing model, we assume that for all links $i,j$, $\calS_i\sim \ln\calN(\mu_i, \sigma^2)$ and $\calI_{ij}\sim \ln\calN(\mu_{ij}, \sigma^2)$ for appropriate positive values $\mu_i,\mu_{ij}$ and $\sigma$. In particular, the second parameter $\sigma$ is constant.

A log-normally distributed variable $X\sim \ln\calN(\mu,\sigma^2)$ can be seen as  $X=e^{Z}$, where $Z\sim \calN(\mu,\sigma^2)$ is a normal random variable. We will use the fact that $\Ex[X]=e^{\mu+\sigma^2/2}$. We will also use the following basic fact.

\begin{fact}\label{F:gaussum}
If $X\sim \calN(\mu_{1},\sigma^2)$ and $Y\sim \calN(\mu_2,\sigma^2)$ are normal random variables, then $\Ex[e^{X-Y}]\le e^{\mu_1-\mu_2+2\sigma^2}$.
\end{fact}
\begin{proof}
By Cauchy-Schwartz, \[\Ex[e^{X-Y}]=\Ex[e^X\cdot e^{-Y}]\le \sqrt{\Ex[e^{2X}]\Ex[e^{-2Y}]}.\] Since $2X\sim \calN(2\mu_{1},4\sigma^2)$ and $-2Y\sim \calN(-2\mu_{2},4\sigma^2)$, using the formula for the expectation of a log-normal variable gives \[\Ex[e^{X-Y}]  \le \sqrt{e^{2\mu_1+2\sigma^2}\cdot e^{-2\mu_2+2\sigma^2}}=e^{\mu_1-\mu_2+2\sigma^2}.\]
\end{proof}

\begin{theorem}
 For any set $L$ of links under  Log-Normal shadowing $\cD$, $\Ex[\OPTS(L)]=\Omega(\OPTG(L))$, even if the signal and interference distributions are arbitrarily correlated.
\end{theorem}
\begin{proof}
Let $S\subseteq L$ be a feasible subset of $L$ under GPL. It is enough to show that the expected size of an optimal feasible subset of $S$ under LNS is $\Omega(|S|)$. Consider an arbitrary link $i\in S$. The affectance of $i$ under LNS is:
\[
\affS_S(i)=\sum_{j\in S\setminus\{i\}} \frac{\SI_{ji}}{\SS_i} = \sum_{j\in S\setminus\{i\}} \frac{e^{Z_{ji}}}{e^{Z_{i}}},
\]
where log-normal random variables $e^{Z_i}$ and $e^{Z_{ji}}$ represent the signal of link $i$ and interference caused by link $j$, respectively. Recall that  
$\GI_{ji}=\Ex[e^{Z_{ji}}]$ and $\GS_i=E[e^{Z_i}]$. Let us denote $\mu_i=E[Z_i]$ and $\mu_{ji}=E[Z_{ji}]$ and note that the variables $Z_i,Z_{ji}$ have variance $\sigma^2$. Hence, using the expectation formula for log-normal variables, we can observe  that $\GS_i = e^{\mu_i + \sigma^2/2}$ and $\GI_{ji}=e^{\mu_{ji} + \sigma^2/2}$. Using this observation together with Fact~\ref{F:gaussum}, we obtain that for every pair $i,j$, 
\begin{align*}
\Ex[\affS_j(i)]&=\Ex[e^{Z_{ji}-Z_i}]\le e^{\mu_{ji} - \mu_i + 2\sigma^2}=\frac{e^{\mu_{ji} + \sigma^2/2}}{e^{\mu_i + \sigma^2/2}}\cdot e^{2\sigma^2}\\
&=\frac{\GI_{ji}}{\GS_i}\cdot e^{2\sigma^2}=e^{2\sigma^2}\affG_j(i).
\end{align*}
Using the fact that $S$ is feasible, and linearity of expectation, we have, for every $i\in S$, that $\Ex[\affS_S(i)]\le e^{2\sigma^2}\cdot \affG_S(i)\le e^{2\sigma^2}/\beta$.
Thus, using  Markov's inequality, we obtain that  
\[
\Pro[\affS_S(i) > 2e^{2\sigma^2}/\beta] <\beta \cdot \Ex[\affS_S(i)]/2e^{2\sigma^2}< 1/2.
\]
The latter implies that the expected number of links $i\in S$ with $\affS_S(i) > 2e^{2\sigma^2}/\beta$ is less than $|S|/2$. It remains to note that by Lemma~\ref{L:robustness}, a $1/(4e^{2\sigma^2})$-th fraction of the remaining set of links will be feasible under LNS, i.e., we will have that $\affS_S(i) < 1/\beta$ for those links $i$.
\end{proof}

\subsection{SS Can Increase Capacity}

Next we show that, perhaps surprisingly, there are instances $L$ for which $\Ex[\OPTS(L)] \gg \OPTG(L)$: the typical optima under SS can be much better than the optimum under GPL. 
The intuition is that shadowing will create many links with higher signal strength than the expectation, which will be the main contributors to the increase in capacity.

In the remainder of this section we consider a set $L$ of links of the same length $l_i=\ell$, and assume that all sender nodes are located at one point (for all) and all receivers are located at another point. We call such links \emph{co-located}. Note that under GPL, any feasible subset of co-located links contains a single link. We show below that under SS, the capacity can significantly increase.
Let us fix a $p$-smooth shadowing distribution $\cD$ for a constant $p>0$. Since $L$ is a set of co-located links of length $\ell$, we have that for each $i,j\in L$, $\GS_i=\GS_j=\GI_{ij}=\GI_{ji}= \frac{P}{\ell^\alpha}$ which we denote $\bar\calS$ for short. 

In a dense or co-located set of links, the only hope for an increase in capacity are links of  signal strength higher than $\bar\calS$. Intuitively, if there is a feasible subset $S\subseteq L$ made of links of strength $k \bar\calS$, then $|S|=O(k)$, since the total interference on each link is likely to be $k\bar\calS$. The following definition essentially captures the maximum size of such a set of strong links there can be in $L$.  

For each link $i$, denote $f(t)=\Pr[\SS_i>t \bar\calS]$ the probability that link $i$ has signal strength at least $t$ times what is expected. For each  integer $n>0$, there is a maximal value $g'_n\ge 0$ such that $f(g'_n) \in \left[\frac{g'_n}{n}, \frac{2g'_n}{n}\right]$, because $f(t)$ is a non-increasing function of $t$, and $f(t)\rightarrow 0$ when $t\rightarrow \infty$. We will use the following slightly different definition: $g_n=\max\{1, g'_n\}$. For the case of log-normal distribution, $g_n=\Theta(e^{2\sigma \sqrt{\ln n}})$, as shown in Cor.~\ref{C:lognormal} below.

The following lemma (Lemma~\ref{L:upperbound}) gives an upper bound on the SS optimum, when the signal strengths of links are fixed and there are few links with ``strong signal''. Essentially, it indicates that the main contribution to the capacity is by links with strong signal, and interestingly, using power control cannot change this.

We will use the following result from linear algebra.
\begin{fact}\label{F:eigen}\cite{Schwenk86,Kolo93}
Let 	$A=(a_{ij})$ be an $n\times n$ non-negative real matrix. Then
$r(A) \ge \frac{1}{n}\sum_{i,j}\sqrt{a_{ij}a_{ji}}$, where $r(A)$ is the largest eigenvalue of $A$.
\end{fact}

\begin{lemma}\label{L:upperbound}
Let $L$ be a set of co-located links. Assume that the signal strengths $\calS_i$ are fixed, but interferences $\SI_{ij}$ are drawn from  a $p$-smooth SS distribution $\cD$, for a number $p>0$. For a number $s > 0$, assume that there are at most $m$ links $i\in S$ with $\calS_i \ge s\cdot \bar\calS$, where $m>\frac{C\max\{s, \log n\}}{ p^2}$ for a large enough constant $C>0$. Then $\OPTS(L) < 4m$ holds w.h.p. with respect to interference distributions, even when links use power control.
\end{lemma}
\begin{proof} Let $m'=4m$.
Let $\calE_{\exists feas}$ denote the event that $L$ contains a feasible (under SS) subset of size $m'$. 
Let $\calE_{S feas}$ denote the event that a set $S$ is feasible. By the union bound, we have that $\Pr[\calE_{\exists feas}] \le \sum_{S\subseteq L, |S|=m'}\Pr[\calE_{S feas}]$. The sum is over ${n \choose m'}\le \left(\frac{en}{m'}\right)^{m'}$ subsets, so it is enough to prove that for each subset $S\subseteq L$ of size $m'$, $\Pr[\calE_{S feas}] < (en)^{-m'}$.

Let us fix a subset $S$ of size $m'=4m$ for the rest of the proof.
Note that $S$ contains a subset $S'$ of $t=3m$ links such that $\calS_i < s\cdot \bar\calS$. Since $\cD$ is a $p$-smooth distribution and all $\SI_{ji}$ are identically distributed, there is a constant $c>0$ s.t. $\bar\calS^p \ge c\bar\calS$ where $\bar \calS^p$ is such that $\Pr[\SI_{ij} > \bar\calS^p]=p$ for every pair $i,j\in L$.
For each pair of links $i,j\in S'$, let $B_{ij}$ denote the binary random variable that is $1$ iff $\min\{\SI_{ij},\SI_{ji}\} > \bar\calS^p$. Note that $B_{ij}$ are i.i.d. variables and $\Pr[B_{ij}=1] = p^2 > 0$. Consider the sum $X=\sum_{i,j\in S'}B_{ij}$. We have that $\Ex[X]=p^2\cdot t(t-1)/2$. By a standard Chernoff bound, we have that $\Pr[X< p^2\cdot t(t-1)/6] < e^{-p^2\cdot t(t-1)/9}$. Recall that $p$ is constant. We choose the constant $C$ so as to have $t\ge 3m \ge 1+\frac{12}{p^2}\cdot \log (en)$, which gives $\Pr[X< p^2\cdot t(t-1)/4]< (en)^{-m}$. 

It remains to prove that
if $X \ge  p^2\cdot t(t-1)/6$, then the set $S$ is \emph{not} feasible with any power assignment.
Let $A=(a_{ij})$ denote the \emph{normalized gain matrix} of the set $S'$, where for any pair of links $i,j\in S'$ we denote $a_{ij}=\SI_{ij}/\calS_j$ and $a_{ii}=1$.
As  shown in \cite{Zander92}, the largest SIR ratio that can be achieved with power control is $\frac{1}{r(A)-1}$, where $r(A)$ is the largest eigenvalue\footnote{It is important here that we defined the normalized gain matrix in terms of interferences and signal strengths with respect to uniform power assignment (i.e., the power level $P$ is ``cancelled'' in the ratio $\SI_{ij}/\calS_j$, leaving the gain ratio, since links $i$ and $j$ use the same power level).} of $A$. Thus, if we show that $r(A) > 1/\beta + 1$, then the set $S$ is not feasible even with power control. To that end, we will use the bound given in Fact~\ref{F:eigen}: $r(A) > \frac{1}{t}\sum_{i,j}\sqrt{a_{ij}a_{ji}}$. We  restrict our attention only to the terms with $B_{ij}=1$, as those will have sufficient contribution to the sum. Indeed, recall that for each pair $i,j$ with $B_{ij}=1$, we have $\min\{\SI_{ij},\SI_{ji}\} > \bar\calS^p \ge c\bar\calS$. Since it also holds for each link $i\in S'$, that $\calS_i < s\cdot \bar\calS$, we have: $\sqrt{a_{ij}a_{ji}}>\sqrt{\frac{(\bar\calS^p)^2}{(s\bar\calS)^2}}\ge c/s$. Since we also have $|\{i,j\in S': B_{ij}=1\}|>p^2 t(t-1)/6$, we obtain the bound:
\[
r(A) > \frac{1}{t}\cdot |\{i,j\in S': B_{ij}=1\}| \cdot \frac{c}{s} \ge \frac{cp^2(t-1)}{6s}.
\]
Hence, for any fixed $\beta$, we choose the constant $C$ so as to have $t\ge 3m>1+\frac{6s(1/\beta+1)}{cp^2}$, in which case $r(A)>1/\beta+1$ and the set $S$ is infeasible.
This completes the proof.
\end{proof}

\begin{theorem}\label{T:basic}
Let $L$ be a set of co-located links under a $p$-smooth SS distribution $\cD$ with associated sequence $\{g_n\}$. There are constants $c_1,c_2>0$, such that 
\[c_1g_n  < \Ex[\OPTS(L)] < c_2 (g_n+\log n),\]
where the expectation is taken w.r.t.\ interference and signal strength distributions.
In particular, if $g_n=\Omega(\log n)$, then $\Ex[\OPTS(L)]=\Theta(g_n)$. 
The upper bound holds even if power control is used.
\end{theorem}

\begin{proof}
 We begin by showing the first inequality. Note that the case $g_n=1$ trivially holds, so we assume $g_n>1$.

Let us call a link \emph{strong} if $\SS_i > g_n \bar\calS$. By the definition of values $g_n$ and since $g_n>1$, the probability that any fixed link $i$ is strong is at least $g_n/n$ and at most $2g_n/n$. Let $S$ be the subset of strong links. The observation above readily implies, by the linearity of expectation, that 
$
g_n\le \Ex[|S|] \le 2g_n.
$

Recall that $\Ex[\SI_{ji}]=\bar\calS$ for all links $i,j\in L$. Thus, the expected affectance on each link $i\in S$ is
\[
\Ex\left[\sum_{j\in S\setminus\{i\}}\frac{\SI_{ji}}{\calS_i}\right]\le \frac{1}{g_n \bar\calS}\Ex\left[\sum_{j\in S\setminus\{i\}}\SI_{ji}\right] = \frac{\bar\calS \cdot\Ex[|S|]}{g_n\bar\calS}\le 2,
\] 
where the expectation is taken over the distribution of interferences $\SI_{ji}$.
Thus, using Markov's inequality, we get that the expected number of links in $S$ with $\affS_S(i) \le 4$  is  $\Omega(|S|)$.  This, together with Lemma~\ref{L:robustness} and the fact that $\Ex[|S|]\ge g_n$, proves that $\Ex[\OPTS(L)] = \Omega(g_n)$.

Now, let us demonstrate the second inequality of the claim. Let us assume that power control is allowed, and, as before, let $\calS_i$ denote the signal strength of link $i$ and $\calI_{ij}$ denote the interference from link $i$ to link $j$ when all links use uniform power $P$. 
Denote $c = c'\cdot(1+\frac{\log n}{g_n})$, where $c'>8$ is a large enough constant. We will show that $\Ex[\OPTS(L)]=O(cg_n)$. 
Let $\calE_{str}$ be the event that there are at least $cg_n/4$ strong links. 
Recall that the expected number of strong links is at  most $2g_n$. Since $c'>8$, applying standard Chernoff bound gives $\Pr[\calE_{str}] < e^{-\log (2n)}<1/n$, i.e., with high probability, $\overline{\calE_{str}}$ holds: There are at most $cg_n/4$ strong links.

 Given that $\overline{\calE_{str}}$ holds, Lemma~\ref{L:upperbound}, applied with $m=cg_n/4$ and $s=g_n$ (assuming the constant $c'$ is suitably large), tells us that gives that with high probability, $\OPTS(L)=O(cg_n)$. This completes the proof.
\end{proof}

\begin{corollary}\label{C:lognormal}
For a set $L$ of $n$ co-located equal length links with Log-Normal Shadowing distribution $\cD$, \[\Ex[\OPTS(L)] = \Theta(e^{2\sigma \sqrt{\ln n}}).\]
\end{corollary}
\begin{proof} 
First, let us estimate $g_n$. Let $\mu$ and $\sigma$ be the parameters associated with the LNS distribution of each link $i$. Since the links have equal lengths, those parameters are the same across all links in $L$.
 Recall that $\bar \calS=\Ex[\SS_i]=e^{\mu+\sigma^2/2}$ for each link $i\in L$.
We will use the fact that the tail probability of a log-normal variable $X$ with parameters $\mu,\sigma$ is as follows:
$
P[X>t] =Q\left(\frac{\ln{t} - \mu}{\sigma}\right),
$
 where $Q(x)=\int_{x}^{\infty}e^{-x^2/2}{d x}$ is the tail probability of the standard normal distribution $\phi(x)=\frac{e^{-x^2/2}}{\sqrt{2\pi}}$. There is no closed form expression for $Q(x)$, but it can be approximated as follows for all $x>0$:
$
\frac{x}{x^2+1}\cdot \phi(x) < Q(x) < \frac{1}{x}\cdot \phi(x).
$

Using these formulas, we obtain:
$
f(t)=\Pro[\SS_i>t\bar \calS]=Q(\ln t/\sigma + \sigma/2).
$
Denote $x=g_n/\sigma + \sigma/2$ for some $n>1$. Then, by the definition of $g_n$, we must have that $f(g_n)=Q(x)=\Theta\left(\frac{e^{-x^2/2}}{x}\right)= \Theta(g_n/n)$. A simplification gives $g_n=\Theta(e^{2\sigma \sqrt{\ln n}})$.
It remains to show that LNS is a smooth distribution, i.e., for some constant $p$, $\bar \calS^p=\Omega(\bar\calS)$. To this end, note that the mean of a log normal variable $X$ with parameters $\sigma,\mu$ is $e^{\mu+\sigma^2/2}$ and the median is $e^\mu$, so assuming $\sigma$ is fixed, we can take $p=1/2$.
\end{proof}

\paragraph{Remark.}  Theorem~\ref{T:basic} can be extended in two ways. First, we can assume that the sender nodes and receiver nodes are not in exactly the same location, but are within a region of diameter smaller compared to the link length. Second, we can assume that the links do not have exactly equal lengths, but the lengths differ by small constant factors. These modifications incur only changes in constant factors. This follows from Lemma~\ref{L:robustness}. That is, the theorem predicts increased capacity due to shadowing not only for co-located sets of links, but also sets that contain dense parts, i.e., have subsets that are ``almost co-located''.

The results above assert that under SS, the main contributor to the capacity increase in a co-located set of links are the ``strong'' links. The following result demonstrates that significant capacity increase (though not as dramatic as above) can happen even when all links have signal strength \emph{fixed} to the GPL value. This phenomenon is due to interference distributions.

\begin{theorem}
Let $L$ be a set of co-located links, all of length $\ell$. Assume that the signal strengths are fixed and equal to $\calS_i=\GS_i=\bar\calS$  but the interferences are drawn from Log-Normal shadowing model $\cD$. Then,
\[
\Ex[\OPTS(L)]=\Omega\left(\sqrt{\log n/\log \log n}\right).
\] 
\end{theorem}
\begin{proof}
Assume, for simplicity, that $\beta=1$.
Consider any fixed subset $S\subseteq L$ of size $|S|=t>e^{\sigma^2}$ for a fixed $t$ to be specified later, and let $i,j\in S$. From the definition of LNS we have,  
 \begin{align*}
 \Pr[\SI_{ij} < \bar\calS/t]&=1-Q\left(\frac{\ln (e^{\mu+\sigma^2/2}/t) - \mu}{\sigma}\right)\\
 &=1-Q\left(\frac{\sigma^2/2-\ln t}{\sigma}\right)\\
 &=Q\left(\frac{\ln t-\sigma^2/2}{\sigma}\right)\\
&\ge \exp(-\ln^2 t/2\sigma^2 +\ln t/2 - \ln \ln t + c),
 \end{align*}
 for a constant $c$, where we used the fact that $Q(x)=1-Q(-x)$. Denote  $\phi(t)=-\ln^2 t/2\sigma^2 +\ln t/2 - \ln \ln t + c$. Thus, the probability that $\SI_{ji} < \bar\calS/t$ holds for all pairs $i,j\in S$, is $e^{t(t-1)\phi(t)}$. Note that if the latter event happens then $S$ is a feasible set. Let us partition $L$ into $\lfloor n/t \rfloor$ subsets of size $t$. From the discussion above we have that each of those subsets is feasible with probability at least $e^{t(t-1)\phi(t)}$. Since those subsets are disjoint, the probability that none of them is feasible is at most 
$(1-e^{-t(t-1)\phi(t)})^{n/t}<\exp(-\frac{n}{t}\cdot e^{-t(t-1)\phi(t)})$. In order to have the latter probability smaller than $1/2$, it is sufficient to set
$n/t > e^{-t(t-1)\phi(t)}$, which holds when $t=O\left(\sqrt{\frac{\log n}{\log \log n}}\right)$
\end{proof}

The following result  demonstrates the limitations of Thm.~\ref{T:basic}. Namely, it shows that the $O(\log n)$ slack in the upper bound cannot be replaced with constant in general: We show that there is an SS distribution with $g_n=O(1)$, such that the gap between SS and GPL capacities is at least a factor of $\Omega(\sqrt{\log n/\log\log\log n})$ for co-located equal length links. It also shows that $g_n=\Omega(\log n)$ is not necessary for increase in capacity compared to GPL.
\begin{theorem}
There is an SS distribution $\cD$ with $g_n=O(1)$ for each $n$, such that 
$\Ex[\OPTS(L)]=\Omega\left(\sqrt{\frac{\log n}{\log\log\log n}}\right)$
 holds for a set $L$ of $n$ co-located  links.
\end{theorem}
\begin{proof}
Let $\calS$ denote the expected signal strength of each link, as before.  Consider the following discrete SS distribution $\cD$, where  for each link $i$ and $t\ge 1$, $\Pr[\SS_i=\frac{c\calS}{2^t}]=\frac{6}{\pi^2 t^2}$, where $c=\left(\frac{6}{\pi^2}\sum_{t\ge 1}\frac{1}{t^22^t}\right)^{-1}$. It is easy to show that $\Ex[\SS_i]=\calS$ and $g_n=O(1)$ for each $n$. Consider an arbitrary subset $S$ of size $m$ and fix a link $i\in S$. It follows from the definition of $\cD$, that  $\Pr[\SS_i \ge \calS] \ge 1/2$. For each link $j\in S\setminus\{i\}$, $\Ex[\SI_{ji}]=\calS$ and $\Pr[\SI_{ij} < \frac{\calS}{\beta \cdot m}]\ge \frac{1}{2\beta \log^2 m}=\frac{1}{2\log^2 m}$. Thus, by independence of interferences, $\Pr[\sum_{j\in S\setminus i}\SI_{ji}<\frac{\calS}{\beta}]\ge \Pr[\SI_{ji}<\frac{\calS}{\beta m}, \mbox{ for all } j\in S\setminus i] \ge \left(\frac{1}{2 \log^2 m}\right)^{m-1}$ and 
\[\Pr[\affS_S(i)<1/\beta]\ge \Pr\left[\sum_{j\in S\setminus i}\SI_{ji}<\frac{\calS}{\beta}|\SS_i\ge \calS\right]\cdot \Pr[\SS_i\ge \calS]\ge \left(\frac{1}{2 \log^2 m}\right)^m.\]
Thus, set $S$ is feasible with probability at least $\left(\frac{1}{2 \log^2 m}\right)^{m^2}$. Split $L$ into $n/m$ disjoint subsets $S_1,S_2,\dots,S_{n/m}$ of size $m$ and let $\calE_t$ denote the event that  $S_t$ is feasible. Note that the events $\calE_t$ are independent, and $\Pr[\calE_t]\ge \left(\frac{1}{2 \log^2 m}\right)^{m^2}$. By the union bound and independence,  $\Pr[\exists t, \calE_t]\ge \frac{n}{m}\cdot \left(\frac{1}{2 \log^2 m}\right)^{m^2}$. Thus, in order to have $\Pr[\exists t, \calE_t]\ge 1/2$, it suffices to take $m\ge \sqrt{\frac{\log n}{3\log\log\log n}}$.
\end{proof}

\section{Computing Capacity under Shadowing}

In this section we study the algorithmic aspect of {\capacity} under SS. Namely, our aim is to design algorithms that
perform well in expectation under SS, compared with the expected optimum under SS. In the first part, we handle (nearly)
co-located links, while the second part discusses the more general case when the links are arbitrarily placed on the
plane. In both cases, we obtain constant factor approximations for links of bounded length diversity,
under general SS distributions satisfying weak technical assumptions.
This holds against an optimum that can use arbitrary power control.

Let us start with several definitions.
For any set $L$ of links, let $\Delta=\Delta(L)=\max_{i,j\in L}\{l_i/l_j\}$ denote the max/min link length ratio in $L$.
A set $L$ is called \emph{equilength} if $\Delta(L) < 2$.

An equilength set $S$ of links is called \emph{cluster} if there is a square of side length $\ell/2$ that contains the sender nodes of all links in $S$, where $\ell=\min_{i\in S}\{l_i\}$.

\subsection{Clusters}\label{S:basicalgo}

Consider a cluster  $L$ of links with minimum length $\ell$. The algorithm is based on Thm.~\ref{T:basic}, which suggests that choosing only the strong links  is sufficient for a constant factor approximation to the expected optimum. We prove that a similar result holds in a more general setting, where
 the signal strengths $\calS_i$ of links $i\in L$ are \emph{fixed and arbitrary}, while interferences are drawn from a $p$-smooth SS distribution $\cD$ for a constant $p>0$. The algorithm is as follows.

\paragraph{Algorithm.} Let set $S$ be constructed by iterating over the set $L$ in a decreasing order of link strength, and adding each link $i$ to $S$ if its strength $\calS_i$ satisfies $\calS_i > 2\beta |S| \bar\calS $. 
Output the set $ALG=\{i\in S: \affS_S(i)< \frac{1}{\beta}\}$ of successful links in $S$.

The theorem below shows that this strategy results in only \emph{additive} $O(\log n)$ expected error, and yields a fully constant factor approximation when signal strengths are drawn from an SS distribution with $g_n=\Omega(\log n)$, such as LNS.

\begin{theorem}\label{T:colocatedalgo}
Let $L$ be a cluster of $n$ links. Assume that the signal strengths of links are fixed and arbitrary, but interferences are drawn from a $p$-smooth SS distribution $\cD$ for a constant $p>0$. Then,
\[
\Ex_{int}[\OPTS(L)] = O(\Ex_{int}[|ALG|] + \log n),
\] where the expectation is taken only w.r.t.\  interference distributions. When the signal strengths are also drawn from $\cD$ and we further have $g_n=\Omega(\log n)$, then 
\[
\Ex[|ALG|]= \Theta(\Ex[\OPTS(L)]),
\]
 where  expectation is taken w.r.t.\ signal and interference distributions.
\end{theorem}
\begin{proof}
First, let us  note that $\Ex[|ALG|]\ge |S|/2$. Indeed, take a link $i\in S$. We have that $\calS_i> 2\beta\bar \calS |S|$ and $\Ex[\sum_{j\in S\setminus i}\SI_{ji}]\le |S|\bar\calS$. Thus, with probability at least $1/2$, $\sum_{j\in S\setminus i}\SI_{ji} \le 2\bar \calS |S| < \calS_i/\beta$, i.e., $i\in ALG$. By additivity of expectation, this implies that $\Ex[|ALG|]\ge |S|/2$.
 It remains to show that $\Ex[\OPTS(L)] = O(|S| + \log n)$. But this simply follows from Lemma~\ref{L:upperbound} with $s=2\beta |S|$ and an appropriate value $m =O(\max\{|S|, \log n\})$. The second part of the theorem follows from Thm.~\ref{T:basic}, which asserts that $\Ex[|S|]=\Theta(g_n)=\Omega(\log n)$.
\end{proof}

\subsection{General Equilength Sets}

The algorithm presented in the  previous section can be extended to  general sets of equilength links, which are not necessarily clusters. The essential idea is to partition such a set into clusters, solve each cluster separately, then combine the solutions. This, however, requires some technical elaboration.

Let $L$ be an arbitrary equilength set. First, we partition $L$ into a constant number of \emph{well-separated subsets}, where an equilength set $S$ is called well-separated if $S$ is a disjoint union of subsets $S_1,S_2,\dots$ such that for each $t$, $S_t$ is a cluster and for each $s\neq t$ and links $i\in S_s$ and $j\in S_t$, $\min\{d(s_i,r_j),d(s_j,r_i)\} > \ell$, where $\ell=\min_{i\in S}\{l_i\}$.

\begin{proposition}\label{P:separatedsets}
Any equilength set $L$ on the plane can be split into a constant number of well-separated subsets.
\end{proposition}
\begin{proof}
Let $\ell=\min_{i\in L}l_i$. Partition the plane into squares of side $\ell/2$ with horizontal and vertical lines. Consider only the squares intersecting the bounding rectangle of links $L$. Assign the squares integer coordinates $(x,y)$ in the following way: if two squares with coordinates $(x,y)$ and $(x',y')$ share a vertical (horizontal) edge then $y=y'$ and $|x-x'|=1$ ($x=x'$ and $|y-y'|=1$, resp.). Then split $L$ into $49$ subsets $L_{ks}$, $k,s=0,1,\dots,6$, where $L_{k,s}=\{i\in L : s_i\text{ is in a square with }(x,y)=(k,s) \mod 7\}$.
Let us fix $R=L_{k,s}$ for some arbitrary indices $s,k$ and let $i,j\in R$ be any links. It remains to note that if the sender nodes of $i$ and $j$ are in different squares, then $d(s_i,s_j)\ge 6\ell/2=3\ell$, which via the triangle inequality (and using the assumption that $L$ is equilength) implies that $d(s_i,r_j),d(s_j,r_i)>\ell$.
\end{proof}

\paragraph{Algorithm.} Partition $L$ into well-separated subsets, solve the capacity problem for each subset separately (as described below) and output the best solution obtained.  

To process a well-separated subset $L'$, observe first that
$L'$ is a disjoint union of clusters $L_1,L_2,\dots,L_m$, by definition. 
Run the algorithm from Sec.~\ref{S:basicalgo} on each cluster $L_t$ separately, obtaining a subset $S_t$. Denote $S=\cup_{t=1}^m S_t$. Let $c>0$ be a constant, as indicated in the proof of Thm.~\ref{T:mainalg}, and let $R$ be the set of all links $j$ in $S$ such that $\affS_S(j) <2c$. Partition $R$ into at most $2c\beta$ feasible  subsets (under SS) using Lemma~\ref{L:robustness} and let $ALG$ be the largest of those.

\begin{theorem}\label{T:mainalg}
Let $L$ be a set of $n$ equilength  links (arbitrarily placed on the plane) under a $p$-smooth SS distribution $\cD$, for a constant $p>0$, with $g_n=\Omega(\log n)$. Then,
\[
\Ex[|ALG|]= \Theta(\Ex[\OPTS(L)]).
\]
\end{theorem}
\begin{proof} 
First, note that performing the step of partitioning into well-separated subsets and selecting the best one, we lose at most a constant factor against the optimum. So we concentrate on a well-separated set $L'$, consisting of clusters $L_1, L_2, \ldots L_m$.
Recall that  $S_1, S_2, \ldots, S_m$ be the subsets obtained by the algorithms on the clusters and $S=\cup_{t=1}^m S_t$. 

We first show that the expected number of links the algorithm chooses from $S$ is $\Omega(|S|)$, where the expectation
is only w.r.t.\ interference distributions. To that end we show that the out-affectance from any link in $S$ to all
the other links is constant under GPL.  That means that the expected affectance on a link in $S$, even under $\cD$, is
constant, which yields the claim after some sparsification of $S$.

For any two indices $s,t\in \{1,\dots,m\}$, $s\ne t$, define the \emph{representative affectance} $a(s,t)$ by set $S_s$ on set $S_t$ as the largest GPL affectance by a link in $S_s$ on a link in $S_t$: $a(s,t)=\max_{i\in S_s,j\in S_t}\affG_i(j)$. It follows from the definition of the GPL affectance and the assumption that all links have length at most $2\ell$, that $a(s,t)=O(1)\cdot (\ell/d_{st})^\alpha$, where $d_{st}=\min_{i\in S_s, j\in S_t}\{d(s_i,r_j)\}\ge \ell$ (by well-separation).
\begin{claim}\label{C:affectbound}
For every $s$, $\sum_{t=1,t\ne s}^m a(s,t)=O(1)$.
\end{claim}
\begin{proof}[Sketch]
This can be shown by an area argument that uses the fact that $\alpha>2$. Take any link $i\in S_s$. Partition the space around the sender node $s_i$ into concentric rings of width $\ell$.  By well-separation, there are no other links within distance $\ell$ from $s_i$. Number the rings starting from $s_i$. Let $a_r$ denote the total contribution of sets $S_t$ that intersect the ring number $r$. It is easy to see that $a_r=O(1/r^\alpha)\cdot z_r$, where $z_r$ is the number of such sets. Since each set $S_t$ occupies an area $O(\ell^2)$ (it is a cluster) and different sets are disjoint, a simple area argument gives that $z_r=O(r)$, so $a_r=O(1/r^{\alpha-2})$. Thus, $\sum_{t=1,t\ne s}^m a(s,t)=\sum_{r=1}^\infty a_r=O(1)$, because $\alpha>2$ and thus the last sum converges. 
\end{proof}

Now, fix an index $s$ and a link $i\in S_s$. Let $j\in S_t$, for $t\ne s$. Recall, from the definition of $S_t$, that $\calS_j = \Omega(|S_t|)\GS_j$. Thus, $\frac{\GI_{ij}}{\calS_j}=O(1/|S_t|)\cdot \affG_i(j)=O(1/|S_t|)\cdot a(s,t)$ and $\sum_{j\in S_t}\frac{\GI_{ij}}{\calS_j}=O(1/|S_t|)\cdot a(s,t)\cdot |S_t|=O(a(s,t))$. A similar argument applied to the links $j\in S_s$ gives that $\frac{\GI_{ij}}{\calS_j}=O(1/|S_s|)\cdot \affG_{i}(j)=O(1/|S_s|)$, so $\sum_{j\in S_s\setminus i}\frac{\GI_{ij}}{\calS_j}=O(1)$.

Summing over all links and using Claim~\ref{C:affectbound}, we obtain that
\begin{align*}
\sum_{j\in S\setminus i}\frac{\GI_{ij}}{\calS_j}&=
\sum_{j\in S_s\setminus i}\frac{\GI_{ij}}{\calS_j} + \sum_{t\neq s}\sum_{j\in S_t}\frac{\GI_{ij}}{\calS_j}\\
&=O(1) + O(1)\cdot \sum_{t\neq s}a(s,t)=O(1).
\end{align*}
Summing over all links $i\in S$, we obtain
$
\sum_{i,j\in S}\frac{\GI_{ij}}{\calS_j}=O(|S|).
$
By the pigeonhole principle, there is a subset $S'\subseteq S$ with $|S'|\ge |S|/2$, such that for each link $j\in S'$, $\sum_{i\in S}\frac{\GI_{ij}}{\calS_j}\le c$ for a constant $c>0$. Thus, for each link $j\in S'$, $\Ex_{int}\left[\sum_{i\in S}\frac{\SI_{ij}}{\calS_j}\right]=\sum_{i\in S}\frac{\GI_{ij}}{\calS_j}<c$. The latter implies that $\Pr[\affS_S(j) < 2c]>1/2$ and the expected number of links $j\in S'$ with $\affS_S(j) < 2c$ is at least $|S'|/2\ge |S|/4$. This is exactly the set $R$ (or $R_r$, for a particular index set $I_r$) we are looking for. Thus, by Lemma~\ref{L:robustness}, $R$ can be split into at most $2\beta c$ feasible subsets. The largest of them, which is output by the algorithm, has expected cardinality at least $|S|/(8\beta c)$.
\smallskip

It now only remains to bound $\OPTS(L)$ in terms of $|S|$. Recall that $\OPTS(L)=O(\OPTS(L'))$. We know from Thm.~\ref{T:colocatedalgo} that 
$
\Ex[|S_t|]=\Theta(\Ex[\OPTS(L_t)])
$
 for $t=1,2,\dots,m$. Combining those observations, we obtain the desired bound:
\begin{align*}
\Ex[\OPTS(L)]&= \Ex[\OPTS(L')] = O\left(\sum_{t=1}^m \OPTS(L_t)\right)\\
&=O\left(\sum_{t=1}^m (\Ex[|S_t|])\right)= O\left(\Ex[|S|]\right).
\end{align*}
\end{proof}

\subsection{General Sets}

Any set of links can be partitioned into at most $\lceil\log\Delta\rceil$ equilength subsets. Thus, we can solve the problem for each of the subsets using the algorithm from the previous section and take the largest solution obtained. That will give us, e.g., $O(\log\Delta)$ approximation for general sets of links under LNS. This holds even against an optimum that can use arbitrary power control.

\section{The Effect of Temporal Fading}\label{S:rayleigh}

In Rayleigh fading model, due to the temporal variability of signals, it is natural to define a "typical optimum" as follows. For a set $L$ of $n$ links (numbered, as before), we consider a transmission probability vector $\bar p=(p_1,p_2,\dots,p_n)$, representing that link $i$ transmits with probability $p_i\in [0,1]$. Let $Q_{\bar p}(i)$ denote the probability that link $i$ succeeds when all links transmit according to the vector $\bar p$ (as defined by the signal to interference ratio). Then, the expected number of successful transmissions is $w(L,\bar p)=\sum_{i\in L} Q_{\bar p}(i)$. Finally, we define the optimum capacity as the maximum expected number of successful transmissions with any transmission probability vector $\bar p$; namely, $\OPTR(L)=\max_{\bar p}w(L,\bar p)$. 

Here we denote the \emph{non-fading} values of signals and interferences as $\calS_i$ and $\calI_{ij}$ and assume that they can take arbitrary positive values. Also, $OPT(L)$ denotes the size of the optimal solution to {\capacity} with respect to those values. We will show that $OPT(L)$ and $\OPTR(L)$ are only constant factor apart from each other, namely, \emph{the effect of temporal fading is at most a (small) constant factor}. 

In fact, we can prove a more general result. Assume that each link $i$ has a positive weight $w_i$. The weighted variant of {\capacity} requires to find a feasible subset $S\subset L$ of links with maximum total weight $\sum_{i\in S} w_i$. We let $WOPT(L)$ denote the weight of an optimal solution in the non-fading setting. Similarly, given a transmission probability vector $\bar p$, let us re-define the expected weight of the set of successful transmissions when using $\bar p$ as $w(L, \bar p)=\sum_{i\in L} w_i Q_{\bar p}(i)$ and let $W\OPTR(L)=\max_{\bar p} w(L, \bar p)$. Then we have the following result. We assume that for every link $i$, the variables $\RS_i$ and $\RI_{ij}$ (for all $j$) are independent. 

\begin{theorem}
For any set $L$ of links, $WOPT(L) = \Theta(W\OPTR(L))$.
\label{lem:main}
\end{theorem}
Dams et al.\ \cite{dams2015} have already shown that $WOPT(L) = O(W\OPTR(L))$. Specifically, given a feasible set $S$, if we form the vector $\bar{p}$ by $p_i = 1$ if $l_i \in S$ and $p_i=0$ otherwise, then $w(L,\bar{p}) \ge |S|/e$.
Thus, it remains to show that $WOPT(L) = \Omega(W\OPTR(L))$.

\begin{proof}
Let $\bar{q}$ be the probability vector corresponding to $W\OPTR(L)$, namely, $\bar q$ is such that $W\OPTR(L) =  w(L,\bar q)=\sum_{i \in L}  w_i Q_{\bar{q}}(i)$.

Note that within the confines of the argument, $\bar{q}$ is a deterministic vector, not a random variable.
Let $A_{\bar{q}}(j) = \sum_{i \in L} q_i \cdot a_i(j)$ denote the \emph{expected affectance} on link $j$.

Dams et al.\ \cite[Lemma 2]{dams2015} applied a characterization of Li and Haenggi \cite{Liu05eurasip}
to obtain that for each link $i$,
\begin{equation}
q_i e^{-A_{\bar{q}}(i)} \le Q_{\bar{q}}(i)) \le q_i e^{-A_{\bar{q}}(i)/2}\ .
\label{eqn:qbnd}
\end{equation}

The idea of our proof is to consider only nodes with small expected affectance under $\bar{q}$; 
they can be sparsified to a feasible set of large size.
The key observation is that much of the weight in the optimal  solution with fading is centered on those low-affectance nodes; 
otherwise, one could obtain a larger solution by uniformly reducing the probabilities.

Let $k = 4 \cdot \ln 11/2 = 4 (\ln 4 + \ln 11/8) \sim 6.819$.
Let $H = \{ i\in L :  A_{\bar{q}}(i) \le k\}$ be the links of small weighted affectance
and let $H' = L \setminus H$.

\begin{claim}
$w(H,\bar{q}) \ge \frac{2}{3} w(L,\bar{q}) = \frac{2}{3}W\OPTR(L)$.
\label{claim:lowaff}
\end{claim}

\begin{proof}
Suppose otherwise, so $w(H',\bar{q}) > \frac{2}{3}w(L,\bar{q})$.
Form a new probability vector $\bar{p} = \bar{q}/4$, i.e., $p_i = q_i/4$, for all $i\in L$.
Then, $A_{\bar{p}}(i) = A_{\bar{q}}(i)/4$, for all $i\in L$.
For $i\in H'$, since $A_{\bar{q}}(i) > k$, it holds from the definition of $k$ that
\begin{equation}
 -\frac{A_{\bar{q}}(i)}{4} - \ln 4 \ge -\frac{A_{\bar{q}}(i)}{2} + \ln 11/8\ . 
\label{eqn:rel}
\end{equation}
Thus, for $i\in H'$, 
\begin{align*} 
Q_{\bar{p}}(i) &\ge \frac{q_i}{4} e^{-A_{\bar{p}}(i)}
  = q_i e^{-A_{\bar{q}}(i)/4 - \ln 4} \\
  &\ge q_i e^{-A_{\bar{q}}(i)/2 + \ln 11/8} 
  = \frac{11}{8} q_i e^{-A_{\bar{q}}(i)/2}
  \ge \frac{11}{8} Q_{\bar{q}}(i)\ , 
	\end{align*}
using (\ref{eqn:qbnd}), the definition of $\bar{p}$, 
(\ref{eqn:rel}), rearrangement, and again (\ref{eqn:qbnd}).
Thus, $w(H',\bar{p}) \ge \frac{11}{8} w(H',\bar{q})$.
It follows that
\begin{align*} 
w(L,\bar{p}) &= w(H,\bar{p}) + w(H',\bar{p}) 
\ge \frac{1}{4} w(H,\bar{q}) + \frac{11}{8} w(H',\bar{q})\\
&= \frac{1}{4} w(L,\bar{q}) + \frac{9}{8} w(H',\bar{q}) >  w(L,\bar{q}), 
\end{align*}
where the strict inequality uses the supposition.
This contradicts the choice of $\bar{q}$. Hence, the claim follows.
\end{proof}

We continue with the proof of Theorem \ref{lem:main}.
We use the probabilistic method to show that there is a feasible subset of $H$ of weight at least $W\OPTR(L)/(6k)$.

Let $X$ be a random set of links from $H$, where link $i\in H$ is added to $X$ with probability $q_i/(2k)$. 
Consider a link $t\in H$ and let $Y_t = a_X(t)$ be the random variable whose value is the affectance of set $X$ on link $t$.
By definition of $H$, $A_{\bar{q}}(t) \le k$.
Observe that 
\begin{align*}
 \Ex[Y_t] &= \Ex[a_X(t)] = \sum_{j \in H} \Pr[j \in X] \cdot a_j(t)\\
  &= \frac{1}{2k} \sum_{j \in H} q_j \cdot a_j(t)
     \le \frac{A_{\bar{q}}(t)}{2k} \le \frac{1}{2} \ , 
		\end{align*}
using the linearity of expectation and the definition of $A_{\bar{q}}$.
Thus, by Markov's inequality, 
\begin{equation}
 \Pr[Y_t \le 1]\, \ge \,\Pr[Y_t \le 2\cdot\Ex[Y_t]] \,\ge \,\frac{1}{2}\ .
\label{eq:markov}
\end{equation}
Observe that the value of $Y_t$ is independent of the event that link $t$ was selected into $X$.
Let $S_X = \{i \in X : a_X(i) \le 1 \}$ be the subset of feasible links in $X$.
It follows that 
\begin{align*}
\Ex[\sum_{i\in S_X}w_i] & = \sum_{i \in L} w_i \Pr[i \in X \text{ and } Y_i \le 1] & \\
   & \ge \sum_{i \in H} w_i \Pr[i \in X] \cdot \Pr[Y_i \le 1]    && \text{(Independence)} \\
   & \ge \sum_{i \in H} w_i \cdot \frac{q_i}{2k} \cdot \frac{1}{2} && \text{((\ref{eq:markov}) and defn.\ of $X$)} \\
   & \ge \frac{1}{4k} \sum_{i \in H} w_i Q_{\bar{q}}(i)              && \text{((\ref{eqn:qbnd}))} \\
   & = \frac{w(H,\bar{q})}{4k} && (\text{Defn.\ of $w$}) \\
   & \ge \frac{W\OPTR(L)}{6k} && \text{(Claim \ref{claim:lowaff}}) \ .
\end{align*}
As per the probabilistic method, this implies there
exists a feasible set of size $\Omega(W\OPTR(L))$.
\end{proof}

\section{Conclusions}

We have compared link capacity in instances with and without stochastic shadowing, as well as with and without temporal fading.
We have also obtained constant-factor approximations in both cases.
Numerous open problems and directions still remain.

There is room to strengthen and generalize our results.
This includes extending the approximations for shadowing to links of unbounded length diversity,
and extending the temporal fading analysis to  other stochastic models.

Modeling correlations and analyzing its effect on link capacity would be valuable,
and the same holds for analyzing other scheduling problems, including weighted link capacity and shortest length schedules.

\mypara{Acknowledgements} This work was supported by grants 152679-05 and 174484-05 from the Icelandic Research Fund.
We thank Michael Neely for pointing out an error in an earlier version of the paper.

\bibliographystyle{abbrv}
\bibliography{Bibliography}

\end{document}